\documentclass[orivec,runningheads]{llncs}
\usepackage[T1]{fontenc}
\usepackage{graphicx}
\usepackage{algpseudocode}
\usepackage{caption}
\usepackage{subcaption}

\usepackage[utf8]{inputenc}

\usepackage{booktabs} 
\usepackage[ruled,noend]{algorithm2e} 
\usepackage{pifont}
\usepackage{nicefrac}
\usepackage{xcolor,colortbl}
\usepackage{graphicx}
\usepackage{amsmath}
\usepackage{amssymb}
\usepackage{bm}
\usepackage{rotating}
\usepackage{multirow}
\usepackage{latexsym}
\usepackage{enumitem}
\usepackage{bm,bbm}
\usepackage{array}

\newcommand{\bv}{\begin{array}}

\newcommand{\ma}{\mathcal A}

\newcommand{\ml}{\mathcal L}

\newcommand{\Omit}[1]{}

\usepackage{hyperref}

\begin{document}
\title{The No-show Paradox in Single Transferable Vote
under One-dimensional Preferences}
\titlerunning{No-show Paradox in STV for 1D Preferences}
%
\author{Farhad Mohsin\orcidID{0000-0003-0224-751X} }
\authorrunning{F. Mohsin}
%
\institute{College of the Holy Cross, Worcester MA 01610, USA
\\
\email{fmohsin@holycross.edu}}
\maketitle              
\begin{abstract}
The group no-show paradox (GNSP) occurs when a group of agents abstaining from voting can make the new winner more preferred to them. Previous work has suggested that even for voting rules susceptible to this paradox, it is a rare occurrence in real elections and under various assumptions. However, we find that under one-dimensional preference models such as 1D-Euclidean, single-peaked, or single-crossing preferences, Single Transferable Vote (STV), a popular runoff rule, is highly vulnerable to GNSP. This is in stark contrast to Condorcet rules, another family of rules susceptible to GNSP, where the paradox cannot occur under these one-dimensional preferences.
We theoretically identify tractable and prevalent sufficient conditions for GNSP to occur for STV under one-dimensional preference models. Through our theoretical results and experiments with synthetic preference profiles from these domains, we demonstrate that voters at the extremes of the 1D spectrum are particularly likely to cause GNSP by abstaining. Furthermore, the likelihood of occurrence increases substantially as the number of alternatives grows.
\end{abstract}
\section{Introduction}

In social choice theory, GNSP refers to the paradoxical event where a group of voters have the incentive to abstain from voting. That is, if they abstain rather than voting truthfully, the winner would be more favorable to them. This was first studied by Fishburn and Brams~\cite{Fishburn1983:Paradoxes}. 
While GNSP does not occur for positional scoring rules like Plurality or Borda, it can unfortunately occur for a wide variety of voting rules including all Condorcet rules and many runoff-rules such as Single Transferable Vote. Existing literature examines the likelihood of GNSP for different voting rules. For example, Mohsin et al.~\cite{mohsin2023computational} showed that verifying whether a preference profile allows for GNSP is an NP-Complete Problem for a variety of voting rules. However, they present Integer Linear Program formulations for solving the problem which work reasonably well in practice for preference profiles with as high as 10 alternatives and 100 voters.
Some work has specifically looked into the likelihood of no-show paradox under specific structured preferences. Felsenthal and Nurmi~\cite{felsenthal2019no} discussed multiple voting rules under a different constraint: when the original preference profile has a Condorcet winner and the voting rule's winner matches the Condorcet winner. Among other results, they suggested that STV is invulnerable to GNSP in this constrained domain. Mohsin et al.~\cite{mohsin2023computational}'s experiments with synthetic preference profiles sampled from different statistical models showed no occurrence of the paradox for in the Single-peaked domain for the following Condorcet rules: Copeland$_{0.5}$, Maximin, and Black's Rule. Contemporarily, Kamwa et al.~\cite{kamwa:hal-03143741} theoretically proved that a variation of GNSP called the Positive Abstention Paradox can not occur for any rule that always chooses the Condorcet winner in the single-peaked domain. They additionally noted that for 3 alternatives, the likelihood of multiple no-show related paradoxes is also reduced for scoring-based runoff rules in the single-peaked domain. 

In this work, we focus on the occurrence of GNSP in restricted domains, in particular those of one-dimensional preferences, such as 1D-Euclidean, single-peaked preferences, and single-crossing preferences~\cite{Elkind2017:structured}. Additionally, there are multiple variations of paradoxes in the family of no-show paradoxes. We follow the focus on the very general scenario where one or more voters (not necessarily identical) strictly benefits by abstaining.

\paragraph{\bf Our contributions:}
\begin{itemize}
    \item We prove that GNSP under this broader definition cannot occur for Condorcet voting rules in single-crossing preferences, extending a known result for single-peaked preferences (Theorem~\ref{thm:sc}).
    \item For STV, we derive concrete sufficient conditions for GNSP to occur under continuous voter distributions and show that these conditions are prevalent under natural assumptions by developing a dynamic programming enumeration algorithm (Propositions~\ref{prop:suff}-\ref{prop:suffsc}).
    \item We formalize the ILP solution for verifying GNSP for STV suggested by Mohsin et al.~\cite{mohsin2023computational} (Section~\ref{sec:ilp}) and using synthetic data experiments show how likelihood of GNSP increases for STV. Through, further exploratory analysis, we show how voters at the ends of the 1D spectrum can be particularly incentivized to abstain due to GNSP (Section~\ref{sec:exp}).
\end{itemize}

\subsection{Related Work}

The no-show paradox, after first being introduced by Fishburn and Brams~\cite{Fishburn1983:Paradoxes}, has been discussed by \cite{Moulin1988:Condorcets,Lepelley2001:Scoring} in slightly differing variations. In this paper, we follow the generalized, group-based definition of \cite{mohsin2023computational}, where the paradox occurs if any group of voters strictly prefers the new election outcome after their collective abstention. In particular, the likelihood of no-show paradox has been studied for a variety of voting rules and different assumption about the voting profiles~\cite{ray1986practical,plassmann2014frequently,Perez2001:The-Strong,duddy2014condorcet}. For this paper, the most relevant are that of \cite{Brandt2021:Exploring} which explored Condorcet rules, \cite{kamwa:hal-03143741,felsenthal2019no}, which explored runoff rules, \cite{brandt2022minimal,mohsin2023computational} which proposed integer linear programs to verify the existence GNSP in order to find occurrences of the paradox.
Additionally, one-dimensional preferences such as single-peaked preferences~\cite{Black48:Rationale}, single-crossing preferences~\cite{Rosenschein94:Rules} etc. have been discussed as structured preference profiles where many voting rules are strategy-proof~\cite{Elkind2017:structured}. We also notice the same for Condorcet rules in both single-peaked and single-crossing preferences, where voters cannot abstain to benefit an alternative they prefer. This makes the high likelihood of the no-show paradox for STV more significant.

\section{Preliminaries}

\subsection{Voting Rules and Paradoxes}

For a set of m alternatives, $\ma$, a ranking is a complete linear order over $\ma$. $\ml(\ma)$ is the set of all rankings, and $|\ml(\ma)| = m!$. If alternative $a$ is preferred over alternative $b$, we say $a \succ b$. For a ranking $R \in \ml(\ma)$, we say $R[a]$ is the position of alternative $a$, and $R^{-1}(j)$ is the alternative at position $j$. For example, for $R = a_1 \succ a_4 \succ a_2 \succ a_3$, $R[a_4] = 2$ and $R^{-1}(3) = a_2$.
A preference profile of $n$ voters, $P \in \ml(\ma)^n$ is a collection of rankings. A {\em voting rule} is a mapping $f: \ml(\ma)^* \rightarrow \ma$ that maps a preference profile of any size to a single winner in $\ma$. To get a unique winner, voting rules use tie-breaking rules. In this paper, we consider the lexicographic tie-breaking rule, which breaks ties in favor of the lexicographically earlier alternative e.g., favoring $a_1$ in a tie between $a_1$ and $a_2$.

Consider the pairwise comparison between each pair of alternatives, $a$ and $b$. If more voters in $P$ prefer $a$ over $b$, we say $a \succ_P b$. If $a \succ_P c$ for all $c \neq a$, then $a$ is called a (strong) Condorcet winner. If an alternative $a$ is undefeated in pairwise comparison, i.e., there is no $c$ such that $c \succ_P a$, then $a$ is a weak Condorcet winner. A voting rule is called a {\em Condorcet Rule} if, whenever a strong Condorcet winner exists, the rule chooses that alternative as the winner. Examples of Condorcet rules are Copeland, Maximin, Black's Rule, Ranked Pairs, among others. For profiles with multiple weak Condorcet winners, tie-breaking rules can be used to choose one of these alternatives as the singular winner. See \cite{Brandt2016:Handbook} for a detailed exposition. 

{\em Positional Scoring Rules} are another family of voting rules that assigns a score to each position. The simplest positional scoring rule is plurality where an alternative's score is the number of top-ranked votes they receive. 
Runoff voting rules are multi-round voting rules that eliminate alternative(s) in each round. Our interest is in {\em Single Transferable Vote (STV)}, also called Instant Runoff Voting and Alternative Voting. STV works by eliminating the alternative with lowest plurality score in each round (breaking ties if necessary), transferring the eliminated alternative's votes to the next ranked alternative, and choosing the alternative remaining at the end of $m-1$ rounds.

{\em The Group No-show Paradox (GNSP)} can formally be defined as the following scenario. For a voting rule, $r$, and preference profile, $P$, assume $f(P) = a$. If there exists an alternative $b$ and a collection of preferences $G \subset P$ such that $b \succ_G a$, and $f(P - G)=b$, that means the voters in $G$ have incentive to abstain from voting as that makes an alternative they prefer over the original winner to become the new winner.

\begin{example}[GNSP Example for STV]\label{ex1}
    \[
    \begin{array}{c}
    2 \  \times a_1 \succ a_2 \succ a_3 \\
    2 \  \times a_2 \succ a_3 \succ a_1 \\
    3 \  \times a_3 \succ a_2 \succ a_1
    \end{array}
    \quad
    \xrightarrow{-1 \times a_1 \succ a_2 \succ a_3}
    \quad
    \begin{array}{c}
    1 \  \times a_1 \succ a_2 \succ a_3 \\
    2 \  \times a_2 \succ a_3 \succ a_1 \\
    3 \  \times a_3 \succ a_2 \succ a_1
    \end{array}
    \]

    The original elimination order was $a_2$ first (breaking ties lexicographically, i.e., in favor of the alternative whose name comes first lexicographically), then $a_1$, and $a_3$ becomes the winner. In the new profile, $a_1$ is eliminated first and $a_3$ gets eliminated next (again, by tie-breaker) and $a_2$ is the new winner, who is preferred more by the abstaining voter, which creates GNSP.
\end{example}

\subsection{One-dimensional Preferences}

A preference profile, $P$, is called 1D-Euclidean if all alternatives and voters can be placed on a straight line. Then, for a voter at $x$, and alternatives $a, b$, if $d(a, x) < d(b, x)$, then $a \succ b$ for $x$, where $d( \cdot )$ is the Euclidean distance function. That is, a voter prefers close alternatives to those farther away.
We also consider two generalizations of 1-Euclidean preference profiles. First, in {\em single peaked preferences}, the alternatives are placed in an axis e.g. $a_1, \rightarrow a_2 \rightarrow \dots \rightarrow a_m$. Then each voter has a most-preferred (peak) alternative, and then alternatives on both sides of the peak have decreasing preferences~\cite{Black48:Rationale}.
Another generalization is {\em single-crossing preferences}, which is based on an ordering of the voters. A preference profile $P$ is single-crossing if the voters are ordered on a single line in a way s.t. for each pair of alternatives $a, b$, there is a point $z$ where all voters to the left of $z$ prefers one of $a$ and all voters to the right prefer $b$ (or vice versa)~\cite{rothstein1991representative}.

\subsection{Integer Linear Programs (ILP)}

ILP is an optimization paradigm to find integer solutions for optimizing a linear objective function under linear and discrete constraints.Beyond standard linear inequalities, ILP formulations can accommodate logical conditions (e.g., conjunctions, disjunctions, and implications) and products of bounded integer and binary variables. While these non-linear expressions can be reformulated into equivalent linear constraints using classic modeling techniques like the big-M method~\cite{bradley1977applied,conforti2014integer}, modern commercial solvers can handle such logical operations and variable products directly. In our formalization for identifying no-show paradoxes under STV, we leverage high-level formulations to maintain conceptual clarity, implementing and solving the resulting models using the Gurobi Optimizer~\cite{gurobi}.

\section{GNSP in One Dimensional Preference Domains}\label{sec:theo}

Upon exploring the existence of GNSP for different voting rules in one-dimensional preferences, we find that Condorcet rules are immune in the domain of one-dimensional preferences. This puts our study of STV in perspective, as some work has suggested that STV can also be robust to GNSP in this domain. In this section, first we discuss existing and new results for GNSP for Condorcet rules, and then discuss STV.

\subsection{GNSP for Condorcet Rules in 1D Domains}

Kamwa et al.~\cite{kamwa:hal-03143741} proved that any voting rule that always selects a (weak) Condorcet winner in the single-peaked domain, cannot suffer any variations of GNSP. The proof their Theorem 2 can be adjusted to the following proposition.
\begin{proposition}
    If $f$ is a Condorcet rule, and $P$ a single-peaked preference profile, then GNSP cannot occur for $f$ in $P$.
\end{proposition}
This comes down to the fact that there is guaranteed to be a (weak) Condorcet winner in a single-peaked profile, that is determined by the median voter(s). And an abstention of voters with preference $b \succ a$ would always push the median voter further away from $b$, hence making GNSP impossible. Similar observations also lead to the non-existence results for GNSP in single-crossing preferences, which we present in the theorem below.

\begin{theorem}\label{thm:sc}
    If $f$ is a Condorcet rule, and $P$ a single-crossing preference profile, then GNSP cannot occur for $f$ in $P$.
\end{theorem}
\proof{
    Let $a = f(P)$ the winner in the original profile. Since $f$ is a Condorcet rule in a single-crossing profile, $a$ must be a (weak) Condorcet winner. So, for any alternative $b \neq a$, we have: $N(a \succ b) \ge N(b \succ a)$ (where $N(a \succ b)$ is the number of voters with $a \succ b$) and any ties are broken in favor of $a$.
    Suppose a non-empty set of voters, $G$, with $b \succ a$ preferences abstain. Then, in the new profile $P - G$, the inequality becomes strict, i.e. $N(a \succ b) > N(b \succ a)$, because there are less number of voters with $b \succ a$ preferences. 
    Now, the single-crossing property is preserved after abstention, so there will still be a Condorcet winner. Since $b$ is defeated in pairwise comparison with $a$, $b$ cannot be a Condorcet winner. Thus, a group of voters with $b \succ a$ preferences cannot make $b$ the new winner by abstaining, and GNSP cannot occur; this completes the proof. \qed
}

\subsection{GNSP for STV in 1D Domains}

Some prior work explores the occurrence of GNSP in single-peaked or other one-dimensional profiles with small number of alternatives (e.g., $m=3$). Some consider more restrictive domains. For example, Felsenthal and Nurmi~\cite{felsenthal2019no} suggested that STV is invulnerable to GNSP when the initial profile has a Condorcet winner that is also the STV winner. Brandt et al.~\cite{brandt2022minimal} proved this to be incorrect by providing a minimal paradox with $n = 9$ voters. Our setting differs from their work slightly because of our consideration of tie-breaking, and we find an even smaller case of GNSP in single-peaked preferences with $n=7$ voters under lexicographic tie-breaking (which served as Example~\ref{ex1}). Since one-dimensional preferences always have a (weak) Condorcet winner, Example~\ref{ex1} provides a smaller paradox for the Felsenthal and Nurmi assumption. 

\subsubsection*{1D-Euclidean and Single-Peaked Preferences}

While Example~\ref{ex1}'s paradox is dependent on the tie-breaking rule, not all no-show paradox instances depend on ties. To study conditions under which no-show paradoxes can occur in one-dimensional preferences, for the next part, we restrict ourselves to 1D-Euclidean preferences. We make the assumption that there is a continuous array of voters in the $[0,1]$ interval, and there are $m$ alternatives, independently and uniformly sampled from $[0,1]$. This departure from finite preference profiles is done for the sake of observing the geometric effects of STV eliminations in each round. We will abuse notation slightly to use $a_i$ to also mean the position of the alternative $a_i$ on the number line.

For STV, in each round, each remaining alternative gets votes from all voters who rank them at the top among all the remaining ones. For our 1D geometric model, this divides the $[0,1]$ segment into nearest neighbor intervals. For example, in round $0$, for $2 \le j \le m-1$, $a_j$'s interval goes from $\frac{a_{j-1}+a_{j}}{2}$ to $\frac{a_{j}+a_{j+1}}{2}$ and has a voting support of $\frac{a_{j+1}-a_{j-1}}{2}$.

\begin{figure}
    \centering
    \includegraphics[width=\linewidth]{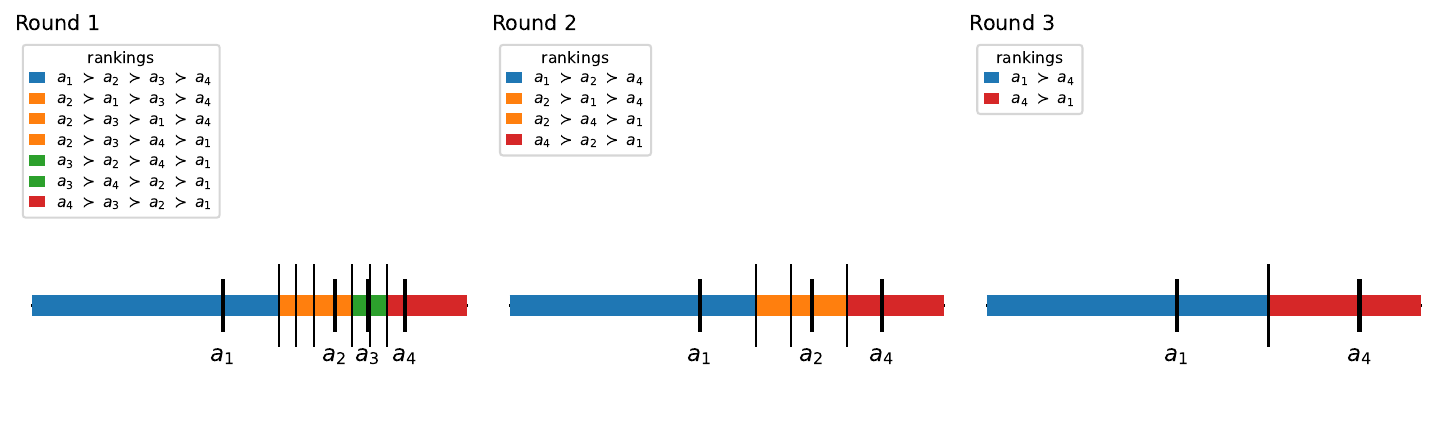}
    \caption{Illustration of STV rounds with contiguous voting regions for each alternative, with votes going to the closest remaining alternative. Note that the long vertical lines are separators for the different rankings, and short vertical lines indicate the positions of $a_1, \dots, a_4$.}
    \label{fig:STV-1d}
\end{figure}

Now, once an alternative is eliminated, in the next round,  their votes must be transferred to the neighboring alternatives because either of the neighbors must be the closest alternative to those voters.\footnote{This observation is also true for single-peaked preferences, where the alternatives are ordered on a fixed axis.} 
In the final round, the last two alternatives $a_i$ and $a_j$ (with $i < j$) divide the electorate at position $\frac{a_i+a_j}{2}$. A left-right separation must occur because no alternative can have non-contiguous voting regions in this 1D model. We call $a_i$ the {\em left consolidating alternative} and $a_j$ the {\em right consolidating alternative}. Alternative $a_i$ wins if $\frac{a_i+a_j}{2} \ge \frac{1}{2}$.

Since each round's elimination only depends on the current plurality score and the eliminations have a consolidating effect for neighboring alternatives, abstentions can have a major effect on the elimination order and the eventual STV winner. For example, imagine if for a profile with $m=5$ alternatives, $a_3$ is the STV winner, and neighboring $a_2$ get eliminated in the first round. But if voters with $a_1 \succ a_2 \succ a_3 \succ a_4 \succ a_5$ abstain from voting, making $a_1$'s plurality score lower than $a_2$, $a_1$ can get eliminated with votes consolidating at $a_2$. Then $a_3$ can get eliminated earlier than $a_2$, and $a_3$ can become the eventual winner, creating a case of GNSP. This might seem like a very special order of events must happen for this to happen, but we will see in Proposition~\ref{prop:suff} how the sufficient condition for this particular event is still rather high.

\begin{proposition}\label{prop:suff}
    Consider a 1D-Euclidean profile with five alternatives $a_1 < \dots < a_5$ and a continuum of voters in $[0,1]$, where $a_3$ is the left-consolidating STV winner, $a_4$ is the right-consolidating alternative, and the profile satisfies one of the following conditions:
    \[
    \begin{aligned}
        \bullet\quad & a_4-a_2 \le a_5 - a_3, \quad \text{or} \\[0.5ex]
        \bullet\quad & 1 - \tfrac{a_4 + a_5}{2} \le \min\left(\tfrac{a_4 - a_2}{2}, \tfrac{a_5 - a_3}{2}\right), \quad \text{and} \quad  \tfrac{a_4 - a_2}{2} < 1 - \tfrac{a_3 + a_4}{2}.
    \end{aligned}
    \]
    If there exists $\ell \in [0, a_1]$ satisfying:
    \[
        \text{(1) } \tfrac{a_3-a_2}{2} \ge a_1 - \ell, \quad
        \text{(2) } a_2 - \ell \ge \tfrac{a_4-a_3}{2}, \quad \text{and} \quad
        \text{(3) } \tfrac{a_2+a_4}{2} \ge \tfrac{1+\ell}{2}.
    \]
    then the abstention of voters in $[0, \ell]$ will change the winner from $a_3$ to $a_2$, creating a no-show paradox.
\end{proposition}
\proof{
We are considering profiles where $a_3$ wins against the right consolidating alternative because of left-consolidation. Assume that either by itself or by consolidating with an early eliminated $a_5$, $a_4$ has higher support than $a_3$ by itself. In the first case, $a_4$'s support is higher than $a_3$'s ($\frac{a_4-a_2}{2} < \frac{a_5-a_3}{2}$). In the second case, $a_5$ has lower support than both $a_3$ and $a_4$ and once eliminated, all $a_5$'s support is consolidated into $a_4$ and $a_4$'s consolidated support is higher than $a_3$. This case indicates that $1-\frac{a_4+a_5}{2} \le \frac{a_4-a_2}{2}$ and $1-\frac{a_4+a_5}{2} \le \frac{a_5-a_3}{2}$. After $a_5$'s elimination, $a_4$'s support becomes $1-\frac{a_3+a_4}{2}$; thus, for $a_3$'s support to be lower: $\frac{a_4-a_2}{2} < 1-\frac{a_3+a_4}{2}$. Together, these comprise the original profile conditions exactly.

Now suppose voters in $[0, \ell]$ abstain, where $\ell \le a_1$. These voters prefer 
$a_1 \succ a_2 \succ a_3 \succ a_4 \succ a_5$, so if $a_2$ wins after their 
abstention, this creates a no-show paradox. The remaining electorate is $[\ell, 1]$ 
with total mass $1-\ell$. Also, note that the abstention does not have any effect on the original support relation between $a_3, a_4, a_5$ that we discussed above.

If the conditions below are satisfied, then $a_1$ can be eliminated before $a_2, a_3$, then $a_3$ will be eliminated before $a_2$, leaving $a_2$ to win against $a_4$ in the final round.
\begin{itemize}
    \item $a_1$ is eliminated before $a_2$, which requires that 
    $\frac{a_3-a_1}{2} \ge \frac{a_1+a_2}{2} - \ell$.
    \item After consolidating all of $a_1$'s voters, $a_2$ has a higher score than $a_3$, guaranteeing that $a_2$ is not eliminated before $a_3$. So 
    $\frac{a_2+a_3}{2} - \ell \ge \frac{a_4-a_2}{2}$. Based on the conditions in the original profile, $a_4$ cannot be eliminated before $a_3$ in this profile with abstentions. So $a_3$ can not consolidate votes from any left or right alternative. Thus, $a_3$ must be eliminated before $a_2$. 
    \item  Next, once $a_3$ is eliminated, $a_2$ must have enough support to become the left-consolidating winner for the remaining $1-\ell$ votes, so 
    $\frac{a_2+a_4}{2} - \ell \ge \frac{1-\ell}{2}$.
\end{itemize}
Thus, if we find any $\ell$ that satisfies these conditions, this constitutes a sequence of elimination where $a_2$ will be the new winner due to abstentions of voters who prefer $a_2 \succ a_3$, creating a case of GNSP. The simplified conditions listed in the proposition statement follow directly from these geometric bounds via straightforward algebraic rearrangement. 
\qed\vspace{2mm}
}

It should be noted that Proposition~\ref{prop:suff} is just one of several paths for how GNSP can occur for STV in 1D preferences. By symmetry, a similar construction can be done for when $a_3$ is a right-consolidating STV winner, and voters with $a_5 \succ \dots \succ a_1$ abstaining can also lead to GNSP. Similar results can be constructed for larger numbers of alternatives.

To assess how frequently these sufficient conditions arise, we conducted simulations sampling $a_1, \ldots, a_5$ uniformly at random from $[0,1]$. For each sample, we find when the original winner condition ($a_3, a_4$ being the final round with $a_3$) occurs. Given this condition, we then check if there exists 
$\ell \in [0, a_1]$ satisfying the conditions in Proposition~\ref{prop:suff}. In $100,000$ samples, $\sim 5.2\%$ satisfied the original winner condition, and out of those, $\sim 19\%$ of the profiles satisfy the sufficient conditions. We further verify that GNSP indeed occurs in all such cases with $a_2$ being the new winner. Also, while this is for $m=5$ alternatives, we can deduce that for any preference profile with $m > 5$ alternatives, after $m-5$ rounds if the remaining alternatives have positions satisfying the conditions in Proposition~\ref{prop:suff}, that means the profile is susceptible to GNSP.
This illustrates that even this one scenario in Proposition~\ref{prop:suff}, which is one of several paths for GNSP, occurs in a significant fraction of random profiles. We will see in Section~\ref{sec:exp} that in one-dimensional preferences with discrete voters, the empirical likelihood of GNSP for STV is indeed quite high, particularly as the number of alternatives increases. We should note that although the preceding discussion revolved around 1D-Euclidean preferences only, the key insight, that when an alternative is eliminated their votes are transferred to the remaining neighboring alternatives, is valid generally for single-peaked preferences, where all the alternatives are ordered in a line.  Section~\ref{sec:exp} shows that GNSP occurs with high frequency in single-peaked profiles as well.

\subsubsection*{Single-Crossing Preferences}

For single-crossing preferences, only the voters are on an axis, and we make the same assumption of a continuum of voters in the $[0,1]$ interval like before. Now, we define a {\em single-crossing domain} as a domain of rankings that start with an identity ranking $a_1 \succ \dots \succ a_m$, repeatedly select two alternatives $a_i,a_j$ that are adjacent in the current ranking and $a_i \succ a_j$ and flip them, and keep doing this till we get to $a_m \succ \dots a_1$. This will require $\frac{m(m-1)}{2}$ steps and will create a domain of rankings where each pair of alternative's preference flips exactly once. The domain can be defined in the $[0,1]$ interval with threshold points for the flips, where the order of the threshold points determine which rankings belong in the domain. 

\begin{figure}[ht]
\begin{subfigure}[h]{0.9\textwidth}
    \centering
    \includegraphics[width=\linewidth]{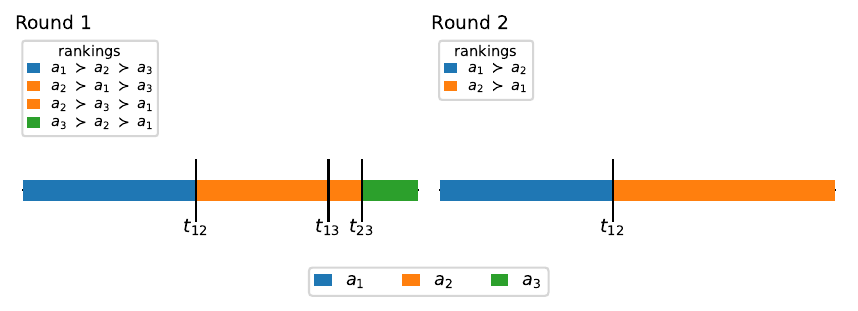}
    \caption{}
    \label{fig:stv-single-crossing}
\end{subfigure}
\begin{subfigure}[h]{0.9\textwidth}
    \centering
    \includegraphics[width=\linewidth]{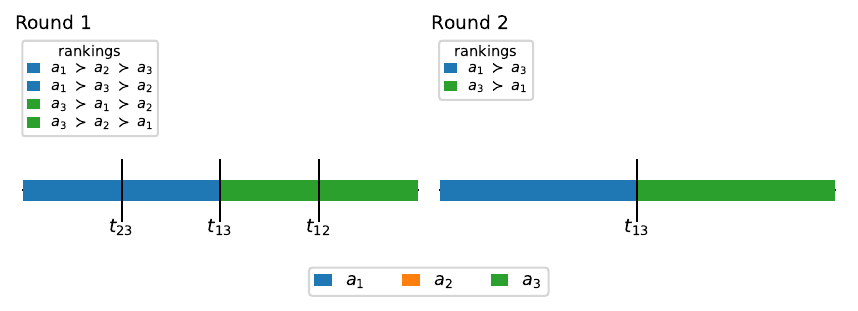}
    \caption{}
    \label{fig:stv-single-crossing-2}
\end{subfigure}
\caption{Showing the top-preference interval for each alternative and the thresholds for two 1D-single-crossing profiles}
\end{figure}

For example, Figure~\ref{fig:stv-single-crossing} shows the case for $m=3$, if we have threshold points $t_{23} < t_{13} \le t_{12}$ where $t_{ij}$ is the flip point for $a_i, a_j$, it leads to a single-crossing domain with the following interval lengths and corresponding voter preferences in $[0,1]$:
\begin{itemize}
    \item $[0, t_{23}]$ of length $t_{23}$ supporting $a_1 \succ a_2 \succ a_3$,
    \item $[t_{23}, t_{13}]$ of length $t_{13} - t_{23}$ supporting $a_1 \succ a_3 \succ a_2$,
    \item $[t_{13}, t_{12}]$ of length $t_{12} - t_{13}$ supporting $a_3 \succ a_1 \succ a_2$,
    \item $[t_{12}, 1]$ of length $1 - t_{12}$ supporting $a_3 \succ a_2 \succ a_1$.
\end{itemize}
Notice that in this domain, $a_1$ is ranked at the top by all voters in the contiguous interval $[0, t_{13}]$, which we will call $a_1$'s {\em top-preference interval}. A key structural feature of single-crossing preference profiles under STV is that an alternative's top-preference interval remains strictly contiguous throughout all the elimination rounds. This contiguity holds for two reasons. First, in the full profile, any alternative $a$’s top-preference region is contiguous because voters can cross a preference threshold between any two alternatives at most once. Thus, once $a$ drops below another alternative in the linear ordering of voters, it can never re-emerge as a top preference. Second, the single-crossing property is perfectly preserved under alternative elimination~\cite{bredereck2013characterization}. Removing an alternative never introduces non-monotonic twists or preference cycles. Thus, the reduced profile at any arbitrary round of STV remains single-crossing, ensuring that the remaining top-preference intervals never fragment into disconnected pieces. Another thing to note is that not all alternatives must have a non-empty top-preference region, as is seen in a different single-crossing domain in Figure~\ref{fig:stv-single-crossing-2}.



\begin{proposition}\label{prop:suffsc}
    Consider a single-crossing profile with alternatives $a_1, \dots, a_m$ and a continuum of voters in $[0,1]$, where $m-3$ alternatives have empty top-preference regions and $a'_1,a'_2,a'_3$ are the competitive alternatives with pairwise thresholds satisfying $t_{12} < t_{13} < t_{23}$. Suppose the original profile satisfies the following conditions:
    \[
    \bullet \quad  t_{23} \le 2t_{12}, \quad 
    \bullet \quad t_{23} < \tfrac{1}{2}(1+t_{12}), \quad \text{and} 
    \quad \bullet \quad t_{13} \ge \tfrac{1}{2}
    \]
    If there exists an abstention interval length $\ell \in [0,1-t_{23}]$ satisfying:
    \[
    \text{(1) } 2t_{23} - t_{12} \ge 1 - \ell, \quad\text{and} \quad \text{(2) } t_{12} < \tfrac{1}{2}(1-\ell)
    \]
    then the abstention of $a'_3 \succ a'_2 \succ a'_1$ voters in $[1 - \ell,1]$ will change the winner from $a'_1$ to $a'_2$, creating a no-show paradox.
\end{proposition}
\proof{
    Since all other $m-3$ alternatives have 0 top preference support, they will be eliminated in the first rounds, leaving the three competitive alternatives. We assumed the pairwise thresholds to have $t_{12} < t_{13} < t_{23}$, thus the original top preference interval lengths for $a'_1,a'_2,a'_3$ are respectively $t_{12}, t_{23}-t_{12}, 1-t_{23}$. The original winner was $a'_1$.
    \begin{itemize}
        \item For the original profile, let $a'_2$ be eliminated first, thus the support top preference interval must be smaller than the others. So, $t_{23} - t_{12} \le t_{12}$ and $t_{23} - t_{12} < 1 - t_{23}$. These lead to the first two conditions.
        \item Originally, after $a'_2$ is eliminated, $a'_1$ wins, so $t_{13} \ge \frac{1}{2}$.
        \item After the abstention of the rightmost voters in $[\ell,1]$, let $a'_3$ be eliminated first because of $a'_2$ having a larger top preference interval. So, $t_{23}-t_{12} \ge 1 - \ell - t_{23}$.
        \item Finally, after $a'_3$'s elimination, let $a'_2$, which has a larger top preference interval now, be the new winner with majority support. Thus, $1-\ell-t_{12} > \frac{1-\ell}{2}$.
    \end{itemize}
    The simplified conditions stated in the proposition then follow via straightforward algebraic rearrangement. Finding any $\ell$ satisfying these bounds thus completes the paradox since the abstaining voters have a preference of $a'_2 \succ a'_1$.
\qed\vspace{2mm}
}

While the proposition depends on a high number of alternatives with zero top-preference. To justify this condition, we enumerate all single-crossing domains starting from a fixed identity permutation $R_1 = a_1 \succ \dots \succ a_m$ based on the structure of permutation reductions. As established previously, a single-crossing domain corresponds to a single maximal chain of adjacent swaps. Each step in the domain flips a single adjacent pair, moving along a path of exactly $\frac{m(m-1)}{2}$ transitions in the maximal chain to get to the completely reversed permutation $R' = a_m \succ \dots \succ a_1$. This is equivalent to the combinatorial problem of choosing a reduced word for the longest element in the symmetric group—until reaching the completely reversed ranking $a_m \succ \dots \succ a_1$~\cite{abello1991:weak,bredereck2013characterization}. While the total number of these maximal chains, thus the total number of possible single-crossing domains, can be analytically counted~\cite{stanley1984number}, we developed a {\em Dynamic Programming} algorithm (detailed in the Appendix) to explicitly traverse this state space. 
This algorithmic approach provides added information about the structural properties of these domains, specifically the distribution of zero top-preference alternatives. We define our sub-problem state as 
$DP(R, \vec{k})$, where $\vec{k} = \langle k_1, \dots, k_m \rangle$, which tracks the number of unique paths leading from the identity to a specific permutation $R$ such that alternative $a_j$ has spent exactly $k_j$ rankings at the top spot. The algorithm terminates and yields the final distribution when $R$ reaches the completely reversed ranking $R'$.
Traversing $DP(R', \vec{k})$, we can compute the number of domains where exactly $m'$ alternatives have $k_j = 0$, i.e., never appear at the top. 
For $m \in \{4, 5, 6, 7\}$, the fractions of domains where exactly $m-3$ alternatives receive zero top-preferences are $0.625$, $0.578$, $0.467$, and $0.359$, respectively. Thus, if we had a preference profile where the support for each ranking in the domain is equal, our condition in Proposition~\ref{prop:suffsc} has quite a high prevalence, showing how STV can be vulnerable to GNSP for single-crossing preferences.

Now, using the same dynamic programming algorithm as above, we can also find domains where an alternative is at the top of more than half of the rankings in a domain. Since each domain has $\frac{m(m-1)}{2} + 1$ unique rankings, because of the $\frac{m(m-1)}{2}$ flips, we need to check $DP(R', \vec{k})$ where
$\max_j k_j \ge  \lceil \frac{1}{2} (\frac{m(m-1)}{2} + 1) \rceil$.
This counts the number of single-crossing domains where alternative $a$ is top-ranked in a majority of the rankings in that domain. If each ranking had equal support, this guarantees $a$ to be the majority winner at the beginning, and in the case that such an alternative exists, GNSP can not occur. Because any alternative that does not have $a$ as  their top choice, when abstaining, increases $a$'s vote share and thus cannot cause GNSP for a plurality-based runoff rule like STV.
For $m \in \{4, 5, 6, 7, 8\}$, the fractions of domains where a single alternative is top-ranked in a majority of the rankings are $0.500$, $0.396$, $0.431$, $0.328$, and $0.200$, respectively. Thus, a large number of domains are less vulnerable to GNSP for single-crossing preferences, reducing the effect. We will see this reflected in the experiments in Section~\ref{sec:exp} where the likelihood of no-show paradox will be lower for single-crossing preferences compared to single-peaked and 1D-Euclidean preferences.

\section{An ILP for finding GNSP in STV}\label{sec:ilp}

Mohsin et al.~\cite{mohsin2023computational} introduced ILP formulations for verifying GNSP, outlining an STV model in their appendix. We provide a complete formalization—with explicit variables, and constraints to ensure reproducibility and support our experimental evaluation (Section~\ref{sec:exp}). This ILP approach is highly efficient for one-dimensional profiles, where the number of unique rankings drops far below the general $m!$ case. Finally, we verified correctness of our STV formulation by confirming its outputs match a brute-force search for minimal paradoxes.

For a preference profile, $P$, let $f_{STV}(P) = a$, be the STV winner. For simplicity's sake, we make the same assumption as \cite{mohsin2023computational} made, of breaking ties using lexicographic tie-breaking. For each alternative $b \neq a$, we need a separate ILP formulation where $b$ becomes the new winner with only voters with $b \succ a$ preferences abstaining. 

\paragraph{\bf Variables} 

\noindent For $m$ alternatives, there can be up to $m!$ unique rankings in a preference profile. Assume the unique rankings are $R_1, R_2, \dots, R_M$ with $M \le m!$. If alternative $c$ is preferred over $d$ in $R_i$, we say $c \succ_{R_i}d$. Let $n_i$ be the number of voters whose vote is $R_i$. Additionally, we define $R_{ab}$ and $R_{ba}$ as the set of rankings with $a \succ b$ and $b \succ a$ respectively. As mentioned before, when $a$ is the original winner, and $b$ is the new (paradoxical) winner, we only consider abstentions from $R_{ba}$.

\noindent For each ranking $R_j \in R_{ba}$, let $x_j \in [0, n_j]$ denote the number of voters with ranking $R_j$ after abstention. These are our primary decision variables. We also require the following auxiliary variables to encode the STV eliminations:

{\footnotesize
\begin{center}
\setlength{\tabcolsep}{8pt} 
\begin{tabular}{l|l|p{6cm}}
Variable & Type & Purpose \\
\midrule
$e(c,r)$ & Binary & 1 if alternative $c$ is eliminated in round $r$ \\
$E(c, r)$ & Binary & 1 if alternative $c$ is eliminated in or before round $r$. $E(c,r) = \sum_{r'=0}^r e(c,r')$. \\
$T(j,k,r)$ & Binary & tracks position eliminations for transfers \\
$STV_{direct}(c)$ & Integer & initial plurality score, counts direct votes  \\
$STV_{transfer}(c,r)$ & Integer & transferred votes \\
$STVScore(c,r)$ & Integer & total score \\
\end{tabular}
\end{center}
}

\noindent Precise definitions are given below with the constraints.

\paragraph{\bf Constraints}

\noindent For STV, we need to consider the score of the remaining alternatives in each round. The score for alternative $c$ at round $1$ is $STV_{direct}(c)$, which counts votes for $c$ as the top choice. For round $r \ge 2$, the score is
\[
STVScore(c,r) = \big(1- E(c, r-1)\big) \big(STV_{direct}(c) + STV_{transfer}(c,r) \big)
\]
where
$STV_{transfer}(c, r)$ counts transferred votes from all $R_i$ where all alternatives $d \succ_{R_i} c$ have already been eliminated. Additionally, the $\big(1- E(c, r-1)\big)$ multiplication ensures a $0$ score for already eliminated alternatives. This product is linearized using the fact that $E(c,r-1)$ is binary.

To track the transfer of votes, define $T(j,k,r)=1$ for $r > 1$, iff all alternatives at positions $1, \dots, k$ in ranking $R_j$, (i.e., $R_j^{-1}[1], \dots, R_j^{-1}[k])$ have been eliminated strictly before round $r$. 
For all rankings $R_j$ and rounds $r$: 
\begin{itemize}
    \item $T(j,1,r) = E[R_j^{-1}[1], r-1]$ and 
    \item we inductively define $T(j,k,r) = \big(T(j, k-1, r) \wedge E(R_j^{-1}[k], r-1) \big)$ for $k \geq 2$.
\end{itemize}

Now, with these auxiliary variables, we can restate $STV_{transfer}(\cdot)$ as follows:
\[
    STV_{transfer}(c, r) = \sum_{k=2}^r \Big( \sum_{\substack{R_j \in R_{ab}\\R_j[c] = k}} n_j T(j,k-1,r) + \sum_{\substack{R_j \in R_{ba}\\R_j[c] = k}} x_j T(j,k-1,r) \;\Big)
\]

Now, the score variable can be tied to the elimination-related variables. 
\[
    \big(e(c,r) = 1 \wedge E(c', r-1) = 0\big) \implies STVScore(c,r) \leq STVScore(c', r)
\]
with strict inequality when $c'$ is lexicographically after $c$ so that $c$ is eliminated according to the lexicographic tie-breaking rule. The implication above is linearized exploiting that both variables in the conjunction are binary.\footnote{This constraint can be modified to account for other tie-breaking rules as well}

Finally, exactly one non-$b$ alternative is eliminated per round: $\sum_{c \neq b} e(c, r) = 1$; each such alternative is eliminated exactly once: $\sum_{r=0}^{m-2} e(c, r) = 1$; and $b$ never is: $\sum_{r} e(b, r) = 0$.

\paragraph{\bf Objective function}

While any feasible solution satisfying all constraints will find an instance of GNSP, we can {\em minimize} the number of abstentions with the following objective:  $\min \sum_{R_j \in R_{ba}} (n_j - x_j)$.

\section{Empirical Results}\label{sec:exp}

This section contains experimental results of two kinds. First, using the ILP described in Section~\ref{sec:ilp}, we calculate empirical likelihood of GNSP for STV under various 1D preference models. Then, we focus specifically on strategic abstentions for voters to see whether a voter or group of voters might have incentive to abstain based on knowledge about positions of the alternatives.\footnote{Code at \url{https://github.com/farhadmohsin/gnsp-stv-1d} with additional details in the Appendix}

\subsection{Likelihood of GNSP in STV}

\paragraph{\bf Sampling Methods} All of our experiments revolve around the one-dimensional distributions we have discussed before.
We consider two well-known methods for sampling single-peaked preference profiles. Conitzer's algorithm~\cite{Conitzer09:Eliciting} samples in a way such that the probability for each alternative being the peak is equal. On the other hand, Walsh's method~\cite{walsh15:generating} samples so that for the chosen order of alternatives, all single-peaked preferences are sampled with equal probability. 
For single-crossing preferences, we follow the method suggested in \cite{Szufa2020:Drawing}. We first create a single-crossing domain but randomly swapping two eligible alternatives at a time until we get to all possible swaps. Then all the votes are sampled uniformly at random from this domain.
For 1D-Euclidean, we fix the $m$ alternatives at $\frac{1}{2m}, \frac{3}{2m}, \dots, \frac{2m-1}{2m}$ and sample $n$ voters, uniformly at random from $[0,1]$. After that, we determine each voter's ranking based on the distance to each alternative. This ensures that the segment in $[0,1]$ where voters could rank an alternative at top is equal for all alternatives. For all of our distributions, we sample 10 batches of 1000 preference profiles, and calculate the empirical likelihood of GNSP occurrence, along with the variance between the batches.

\begin{figure}[ht]
    \centering
    \includegraphics[width=\linewidth]{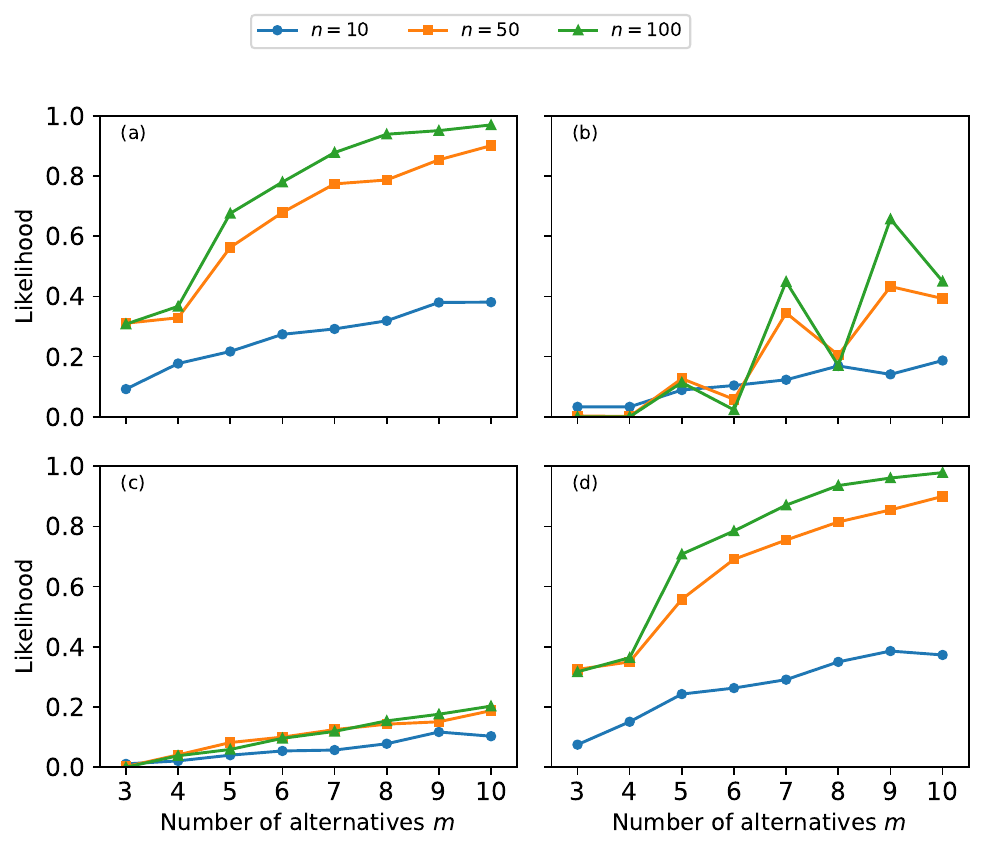}
    \caption{Likelihood of GNSP for STV for increasing number of alternatives (a) Single-Peaked (Conitzer), (b) Single-Peaked (Walsh), (c) Single-Crossing, (d) 1D-Euclidean}
    \label{fig:likelihood}
\end{figure}

Figure~\ref{fig:likelihood} shows how the likelihood of GNSP changes for different distributions. For all distributions, we generally see the likelihood increase with an increase in either the number of voters or the number of alternatives. In particular, for single-peaked preferences, when sampled with Conitzer's method or the 1D-Euclidean method, we see the empirical likelihood go up to as high as 95\% by $m=9$.\footnote{See the full tables along with error margins in the Appendix.} This highlights the high risk of GNSP for this type of profiles. One major difference in the preference profiles sampled from these two sampling methods is that, in order to give all the alternatives an equal chance of being the peak, Conitzer's method samples more votes at the two extremes ($a_1 \succ \dots \succ a_m$ and its reverse), and these voters are  highly likely to have an effect on the outcome by abstaining, as exemplified by Proposition~\ref{prop:suff}. In Figure~\ref{fig:heatmap5}(a), we see a heatmap of how the winner changes due to abstention in cases of GNSP for single-peaked profiles. The high density towards the center ($a_2,a_4 \rightarrow a_3, a_3 \rightarrow a_2,a_4$) are almost entirely caused by voters at either spectrum abstaining. For example, $a_1 \succ \dots \succ a_5$ voters (and their close neighbors) abstaining will have the effect of changing the winner from $a_3$ to $a_2$ or $a_4$ to $a_3$. And this effect is similar for higher number of alternatives as well. We show the heatmap for $m = 9$ in the Appendix. For single-peaked preferences sampled with Walsh's method, we see that the likelihoods are non-monotonic in $m$, particularly fluctuating by parity, but the overall upward trend in GNSP likelihood is preserved within each parity class. We give an informal discussion, rooted in how Walsh sampling concentrates STV winners near the center and how for odd-vs-even number of alternatives, there can be one-vs-two central alternatives, in the Appendix.

Conversely, for single-crossing preferences in Figure~\ref{fig:likelihood}(c), while there is an increasing fraction of cases of GNSP, the rate is comparatively lower than for single-peaking preferences. This aligns with our observation that there is a large number of domains of single-crossing preferences where GNSP is unlikely to occur. The heatmap at Figure~\ref{fig:heatmap5}(b) indicates how the winners are changing from $a_1$ or $a_5$ to one of the inner alternatives. In Proposition~\ref{prop:suffsc}, we showed that a likely transition in case of GNSP is that the alternative in the middle will be the new winner due to abstentions, which matches with the observations in the heatmap.

\begin{figure}[t]
    \centering
    \includegraphics[width=\linewidth]{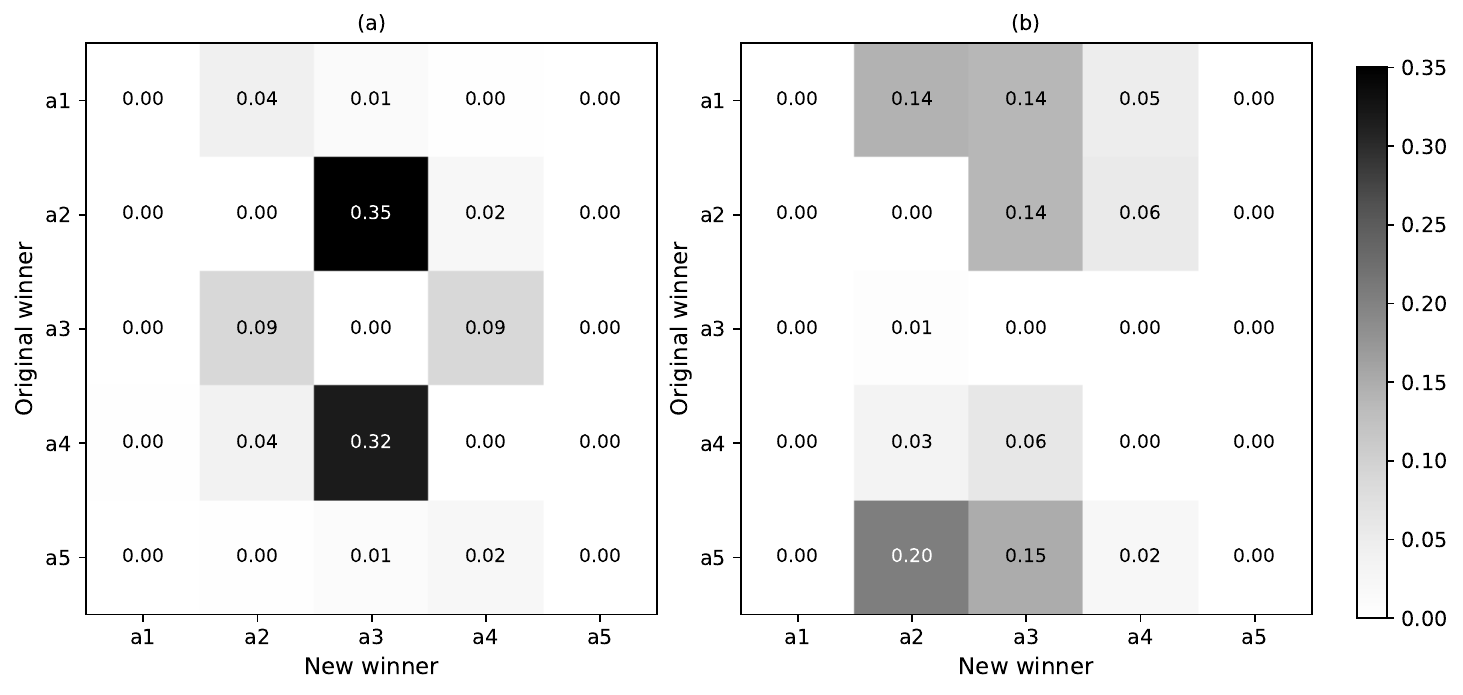}
    \caption{Likelihood heatmap of winner transitions caused by GNSP for STV. (a) Single-peaked preferences (Conitzer), (b) Single-crossing preferences}
    \label{fig:heatmap5}
\end{figure}

\setlength{\tabcolsep}{6pt} 

\begin{table}[htbp]
\centering
\caption{Fractions of profiles (out of profiles where $k$ abstentions of the type is possible) by voter type and abstention size for $m=5$, where group is better/worse. \textbf{Bold} indicates scenarios where abstaining is highly incentivized.}
\label{tab:abstention}
\begin{tabular}{l*{6}{c}}
& \multicolumn{6}{c}{Abstention Size} \\
\cmidrule(lr){2-7}
& \multicolumn{2}{c}{$k=1$} & \multicolumn{2}{c}{$k=2$} & \multicolumn{2}{c}{$k=3$} \\
\cmidrule(lr){2-3} \cmidrule(lr){4-5} \cmidrule(lr){6-7}
Voter Type & Better & Worse & Better & Worse & Better & Worse \\
\midrule
$a_1 \succ a_2 \succ a_3 \succ a_4 \succ a_5$ & 0.025 & 0.049 & 0.047 & 0.059 & 0.068 & 0.102 \\
$a_2 \succ a_1 \succ a_3 \succ a_4 \succ a_5$ & \textbf{0.055} & 0.009 & \textbf{0.099} & 0.009 & \textbf{0.141} & 0.012 \\
$a_2 \succ a_3 \succ a_1 \succ a_4 \succ a_5$ & 0.026 & 0.026 & 0.049 & 0.047 & 0.058 & 0.077 \\
$a_3 \succ a_2 \succ a_4 \succ a_1 \succ a_5$ & 0.009 & 0.071 & 0.014 & 0.097 & 0.020 & 0.166 \\
$a_3 \succ a_4 \succ a_2 \succ a_5 \succ a_1$ & 0.008 & 0.009 & 0.014 & 0.074 & 0.019 & 0.088 \\
$a_4 \succ a_3 \succ a_5 \succ a_2 \succ a_1$ & 0.028 & 0.019 & 0.045 & 0.045 & 0.059 & 0.063 \\
$a_4 \succ a_5 \succ a_3 \succ a_2 \succ a_1$ & \textbf{0.052} & 0.000 & \textbf{0.095} & 0.000 & \textbf{0.140} & 0.001 \\
$a_5 \succ a_4 \succ a_3 \succ a_2 \succ a_1$ & 0.018 & 0.014 & 0.037 & 0.064 & 0.062 & 0.079 \\
\bottomrule
\end{tabular}
\end{table}

\subsection{Strategic Abstention in STV}

GNSP, while undesirable, has generally been considered low-risk because of low likelihood for most voting rules. But this work identifies one domain (that of one-dimensional preferences) where the risk is rather high. This leads us to explore whether voters can think of strategically abstaining even without full information about all other voters. In this section, we perform a brief experiment with simulated data with 1D-Euclidean preferences. We sample 10,000 preference profiles with the fixed alternative locations, and for individual ranking type, we observe the effect of abstentions. We consider three type of outcomes, no change to winner, new winner more preferred, and new winner less preferred. And we find that for $m \in \{4,5,6\}$, there are always voter types, for whom the fraction of profiles where new the winner is more preferred is higher than the profiles where the new winner is less preferred. For example, for $m=5$, if a group of $a_2 \succ a_1 \succ a_3 \succ a_4 \succ a_5$ abstain in group sizes of $k \in \{1,2,3\}$, the group is up to ten times more likely to be strictly better off as opposed to strictly worse (in the vast majority, the abstention still has no effect). See Table~\ref{tab:abstention} for details (more results in the Appendix). This indicates that knowledge of alternative positions, along with knowledge with distribution of the voters 
can be enough to give voters a good estimate of their expected benefit from abstaining.

\section{Discussion and Future Work}

Our theoretical and empirical work show that, STV, a popular choice of voting rule, suffers highly from the risk of the no-show paradox in structured one-dimensional settings. This risk is particularly high under single-peaked preferences, when the extreme positions at opposite ends of the spectrum maintain a high likelihood of attracting voter peaks. As our simulations demonstrate, the empirical likelihood of GNSP can go beyond $95\%$ for 8–10 alternatives. This is made worse by the fact that with knowledge about the alternative positions and distribution of voters, it can be computationally easy for some voters to find probabilistic justification to abstain from voting to get a strictly preferred outcome.
Future work can focus on a full game-theoretic characterization of individual strategic abstention under partial or noisy information. 

\newpage

\begin{credits}
\subsubsection{\ackname} We thank Lirong Xia for continued discussions on no-show paradoxes, and Zack Fitzsimmons for valuable feedback on an earlier version, all of which greatly improved the logic and readability of the paper.

\end{credits}
%
%
%
\bibliographystyle{splncs04}
\bibliography{new_ref,references}

\newpage
\appendix 

\section{Dynamic Programming Algorithm for Single-Crossing Domains}

To create any single-crossing domain, we always start with the same identity ranking, $R_1 = a_1 \succ \dots \succ a_m$. At each step for the current ranking $R$, pick two adjacent alternatives at positions $i, i+1$ such that $R[i]$ is lexicographically before $R[i+1]$. In that case, we can flip their positions to get to the next ranking, $P$. And we keep adding more rankings to the domain until we reach the reverse of the identity, $R' = a_m \succ a_{m-1} \succ \dots \succ a_1$.

Our goal is to tabulate $DP(R, \vec{k})$ where $\vec{k} = \langle k_1, \dots, k_m \rangle$, which tracks the number of unique paths leading from the identity to a specific permutation $R$ such that alternative $a_j$ has spent exactly $k_j$ rankings at the top spot. 

Define $\hat{e_i}^{(m)}$ to be the unit vector of length $m$ where the $i$-th entry is $1$ and everything else is $0$.
Naturally, we can initialize $DP(R_1, \hat{e_1}^{(m)}) = 1$ since $a_1$ is the element in top in $R_1$. From there, every time, we perform a flip to get from ranking $R_i$ to the next ranking $R_j$, we get:
\[
DP(R_j,\ \vec{k} + \hat{e}_i^{(m)}) = DP(R_j,\ \vec{k} +
  \hat{e}_i^{(m)}) + DP(R_i,\ \vec{k}),
\]
where $a_i = R_j^{-1}(1)$, the top-ranked alternative of $R_j$.

This DP algorithm is better than a brute force approach would have been because many chains collapse to the same $(R, \vec{k})$ chain. However, the number of possible single-crossing domains still grows super-exponentially fast, so we could only run our algorithm for up to $m=8$, which nonetheless showed us the prevalence we were looking for and which we reported in the body of the paper.

\section{Likelihood of GNSP-vs-m-vs-n tables}

\begin{table}[htbp]
\centering
\caption{GNSP Likelihood For Single-Peaked (Conitzer)}
\begin{tabular}{c|c|c|c}
$m$ & $N=10$ & $N=50$ & $N=100$ \\
\midrule
3  & $9.2 \pm 2.7$   & $31.1 \pm 4.93$ & $30.8 \pm 3.61$ \\
4  & $17.7 \pm 4.42$ & $32.9 \pm 3.57$ & $36.7 \pm 5.7$  \\
5  & $21.7 \pm 3.2$  & $56.3 \pm 5.1$  & $67.6 \pm 2.8$  \\
6  & $27.4 \pm 4.35$ & $67.9 \pm 5.67$ & $78 \pm 5.27$   \\
7  & $29.2 \pm 3.12$ & $77.4 \pm 4.25$ & $87.8 \pm 4.24$ \\
8  & $31.9 \pm 4.51$ & $78.7 \pm 3.92$ & $93.9 \pm 2.13$ \\
9  & $38 \pm 4.27$   & $85.4 \pm 2.76$ & $95.1 \pm 2.23$ \\
10 & $38.1 \pm 4.58$ & $90.1 \pm 3.57$ & $97 \pm 1.94$   \\
\end{tabular}
\end{table}

\begin{table}[htbp]
\centering
\caption{GNSP likelihood for Single-Peaked (Walsh)}
\begin{tabular}{c|c|c|c}
$m$ & $N=10$ & $N=50$ & $N=100$ \\
\midrule
3  & $3.3 \pm 1.7$   & $0.3 \pm 0.67$  & $0 \pm 0$        \\
4  & $3.3 \pm 1.64$  & $0.1 \pm 0.32$  & $0 \pm 0$        \\
5  & $8.9 \pm 2.81$  & $12.7 \pm 3.83$ & $11.4 \pm 2.95$  \\
6  & $10.4 \pm 2.8$  & $5.8 \pm 2.25$  & $2.3 \pm 1.49$   \\
7  & $12.3 \pm 2.5$  & $34.5 \pm 1.51$ & $44.9 \pm 4.63$  \\
8  & $16.9 \pm 2.92$ & $20.5 \pm 2.51$ & $17 \pm 4.52$    \\
9  & $14.1 \pm 3.38$ & $43.3 \pm 5.62$ & $65.7 \pm 5.12$  \\
10 & $18.7 \pm 4.37$ & $39.3 \pm 3.95$ & $45 \pm 4.27$    \\
\end{tabular}
\end{table}

\begin{table}[htbp]
\centering
\caption{GNSP likelihood for Single-Crossing}
\begin{tabular}{c|c|c|c}
$m$ & $N=10$ & $N=50$ & $N=100$ \\
\midrule
3  & $1.1 \pm 0.88$ & $0.2 \pm 0.42$  & $0 \pm 0$        \\
4  & $2.1 \pm 1.85$ & $4.1 \pm 1.2$   & $3.8 \pm 2.2$    \\
5  & $4 \pm 1.83$   & $8.2 \pm 2.39$  & $5.9 \pm 1.97$   \\
6  & $5.4 \pm 2.46$ & $10 \pm 2.62$   & $9.6 \pm 2.84$   \\
7  & $5.7 \pm 2.21$ & $12.5 \pm 3.89$ & $11.9 \pm 2.69$  \\
8  & $7.8 \pm 2.39$ & $14.3 \pm 4.99$ & $15.4 \pm 3.72$  \\
9  & $11.7 \pm 3.3$ & $15.1 \pm 2.92$ & $17.6 \pm 2.95$  \\
10 & $10.3 \pm 2.31$ & $18.8 \pm 4.66$ & $20.3 \pm 4.08$ \\
\end{tabular}
\end{table}

\begin{table}[htbp]
\centering
\caption{GNSP likelihood for 1D-Euclidean}
\begin{tabular}{c|c|c|c}
$m$ & $N=10$ & $N=50$ & $N=100$ \\
\midrule
3  & $7.5 \pm 2.76$  & $32.5 \pm 2.99$ & $31.7 \pm 4.5$  \\
4  & $15.1 \pm 3.87$ & $35 \pm 4.52$   & $36.4 \pm 4.65$ \\
5  & $24.3 \pm 4.35$ & $55.8 \pm 4.47$ & $70.8 \pm 4.18$ \\
6  & $26.3 \pm 3.23$ & $69.1 \pm 5.07$ & $78.5 \pm 4.97$ \\
7  & $29.1 \pm 5.36$ & $75.5 \pm 4.58$ & $87.1 \pm 2.69$ \\
8  & $35 \pm 5.54$   & $81.5 \pm 4.5$  & $93.6 \pm 2.76$ \\
9  & $38.6 \pm 3.24$ & $85.5 \pm 3.27$ & $96.1 \pm 1.66$ \\
10 & $37.3 \pm 5.89$ & $90 \pm 1.83$   & $97.9 \pm 1.45$ \\
\end{tabular}
\end{table}

\newpage
\section{Heatmap for GNSP transitions for $m=8$ and $m=9$}

\begin{figure}[ht]
\begin{subfigure}[h]{\textwidth}
    \centering
    \includegraphics[width=\linewidth]{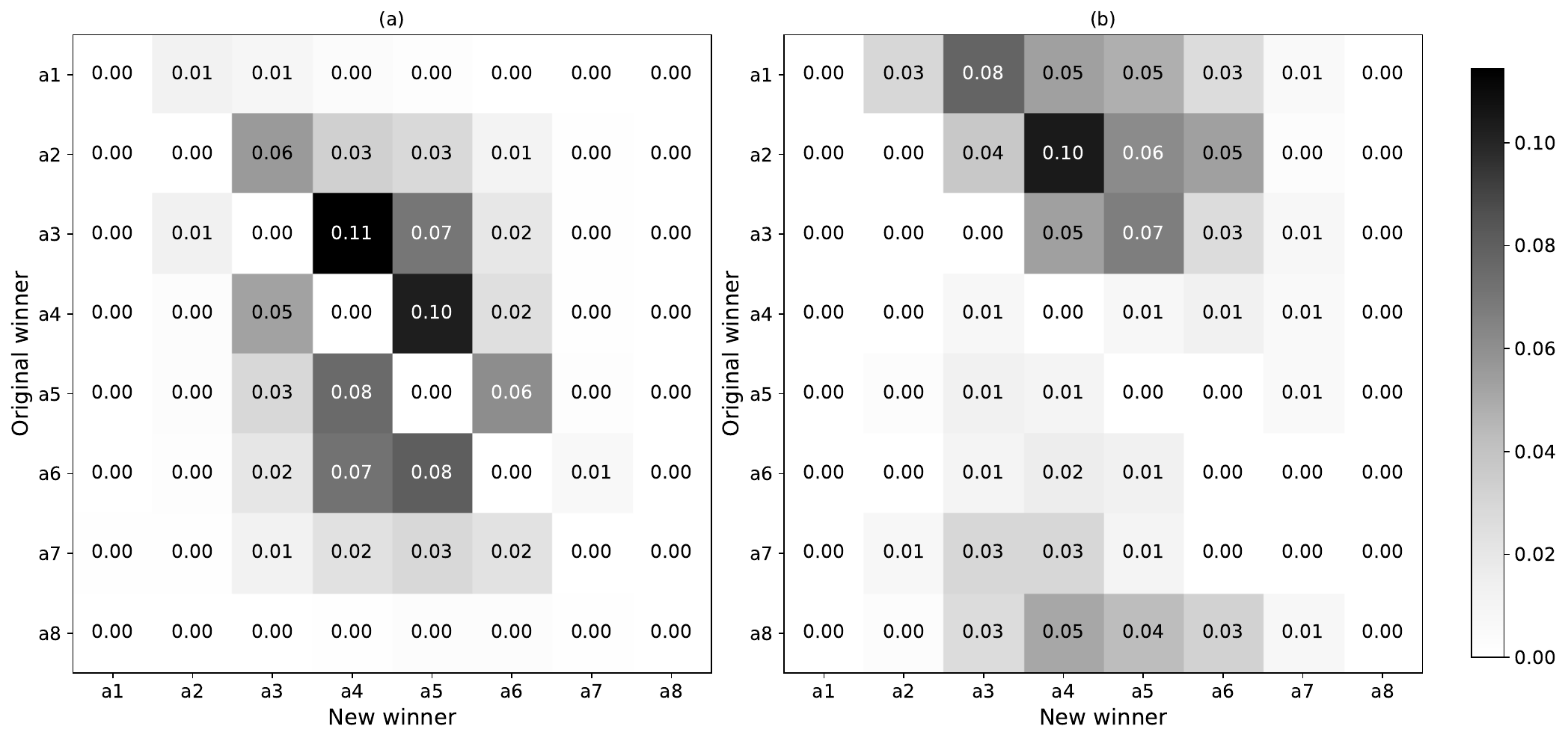}
    \label{fig:heatmap8}
\end{subfigure}
\begin{subfigure}[h]{\textwidth}
    \centering
    \includegraphics[width=\linewidth]{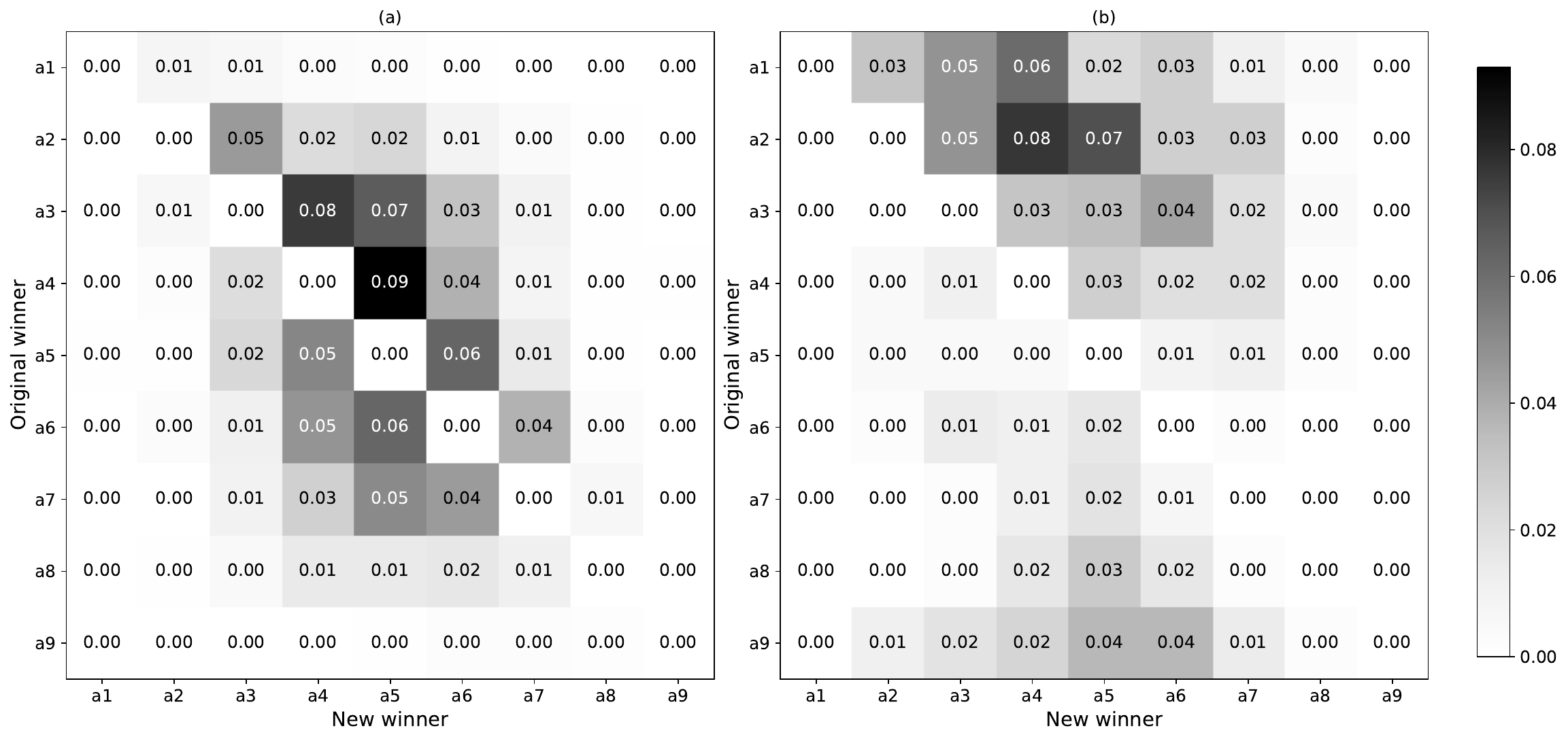}
    \label{fig:heatmap9}
\end{subfigure}
\caption{Likelihood heatmap of winner transitions caused by GNSP for STV for $m=8$. (left) Single-peaked preferences (Conitzer), (right) Single-crossing preferences. (top) $m=8$, (bottom) $m=9$}
\end{figure}

\newpage

\section{Results for Single-peaked Preferences sampled with Walsh's method}

As noted, the empirical likelihood for single-peaked preferences sampled with Walsh's method does not grow monotonically, like the other sampling profiles. It specifically shows an odd-even pattern. 
We first note that in contrast to the profiles generated with Conitzer's method, for Walsh's method, since all single-peaked rankings (under a certain ordering of the alternatives) are equally likely, the winners tend to be at center, as seen in Figure~\ref{fig:hist_spw}. Particularly for odd number of alternatives, the median alternative is the winner the highest number of times, with its two neighboring alternative with some ratio. On the other hand, for even number of alternatives (exemplified with $m=8$), there are two median alternatives, both of which are highly likely to be the winner. 
Hence, for GNSP to occur, the winner needs to change from (one of) the median alternative(s), to one towards the edge. 

\begin{figure}[ht]
    \centering
    \includegraphics[width=\linewidth]{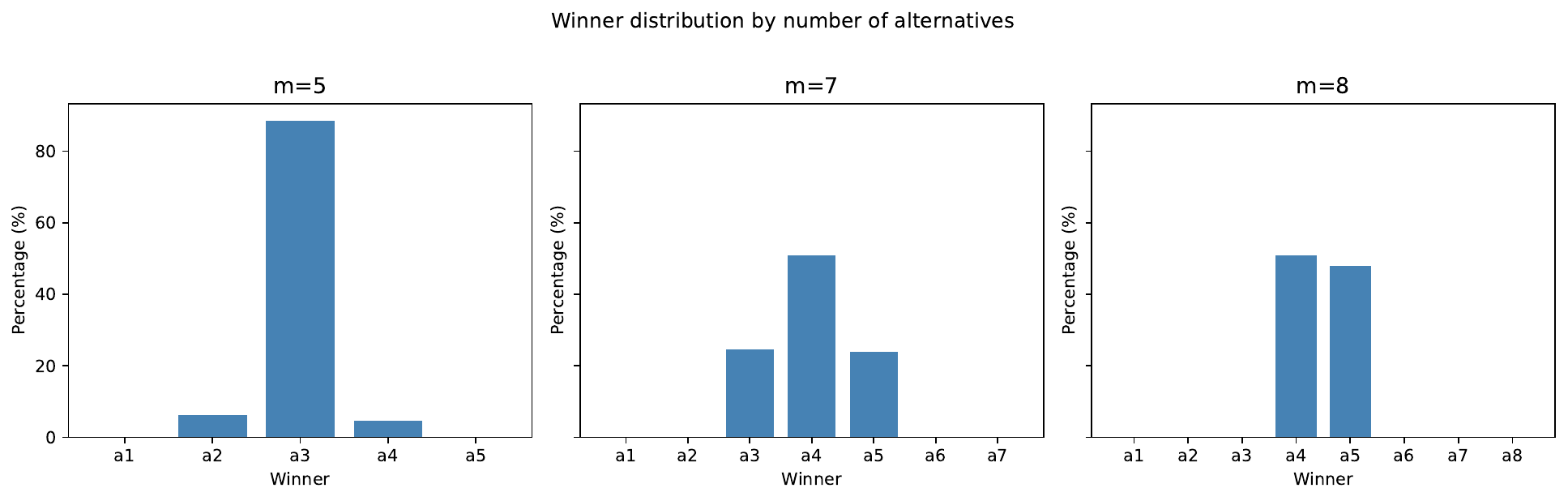}
    \caption{STV winner empirical distribution for single-peaked preferences (Walsh-sampled)}
    \label{fig:hist_spw}
\end{figure}

\begin{figure}[ht]
    \centering
    \includegraphics[width=\linewidth]{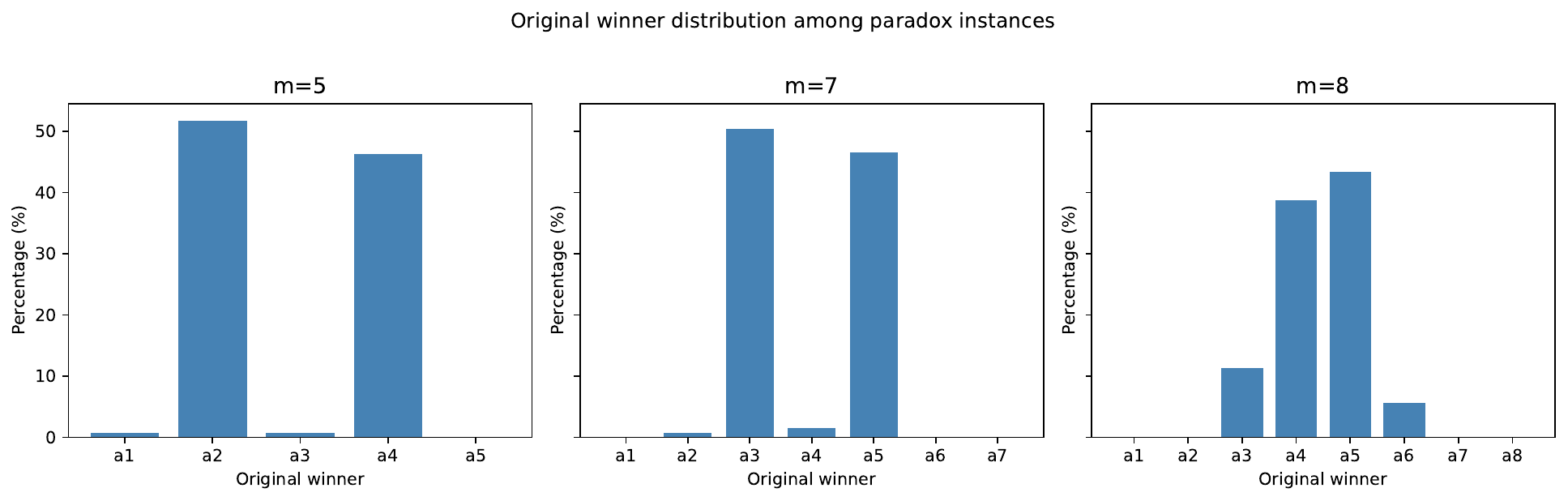}
    \caption{STV winner empirical distribution for single-peaked preferences (Walsh-sampled) where GNSP occurs}
    \label{fig:hist_gnsp_spw}
\end{figure}

\begin{figure}[ht]
    \centering
    \includegraphics[width=\linewidth]{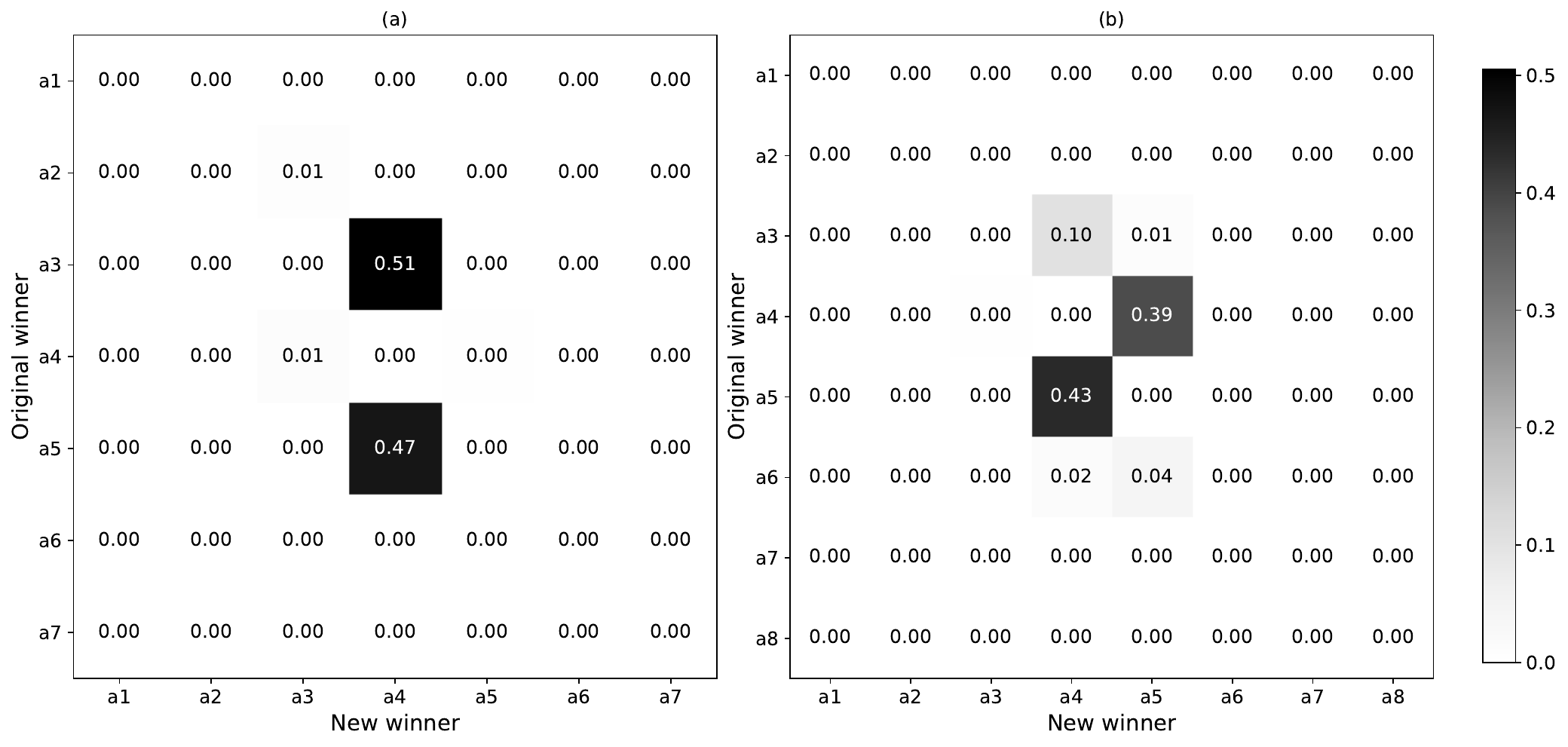}
    \caption{Likelihood heatmap of winner transitions caused by GNSP for STV for single-peaked preferences (Walsh). (Left) $m=7$, (Right) $m=8$}
    \label{fig:heatmap_spw}
\end{figure}

In Figures~\ref{fig:hist_gnsp_spw} and~\ref{fig:heatmap_spw}, we observe the histogram of original winner and heatmap for original to new winner for profiles with GNSP. For odd number of alternatives, we notice a high likelihood of GNSP for when the original winner was one of the neighbors to the median alternative. And then the GNSP causes the winner to move towards the median alternative. We notice that this towards-center move is reminiscent of Proposition~\ref{prop:suff} which was for 1D-Euclidean profiles, even though there are no explicitly left or right voters for single-peaked preferences. For even alternative profiles, we see that the bulk of GNSP occur when the winner changes from one median alternative to the other one. We notice that as we go to more alternatives ($m=5$ to $m=7$), the likelihood of winning for the median-neighbors increase and along with that the empirical likelihood increases as well.

For even alternative profiles, we note the GNSP occurs mostly when the winner changes from one median to another. We conjecture that since each alternative is very robust because of how the profiles are sampled, there is less mass that can pull the winner farther away from the center. The likelihood pattern for single-peaked (Walsh) preferences thus shows us an interesting pattern and definitely a problem that can be addressed in future work. We also think that on a similar note, there is work that can be done to find sufficient conditions for single-peaked preferences in general like we do for 1D-Euclidean and single-crossing preferences.

\newpage

\section{Experimental Results for Strategic Abstention}

For $m \in \{4,5,6\}$, we run simulations with synthetic 1D-Euclidean preference profiles. We fix one type of voter (so, a specific ranking) at a time, and choose group sizes of $k \in \{1,2,3\}$ to abstain from voting. We report the fraction of profiles where the (group of) voters are strictly better off for abstaining as opposed to the fraction where they are strictly worse.

\begin{table}[ht]
\centering
\caption{Fractions of profiles better/worse by Voter Type and Abstention Size for $m=4$. \textbf{Bold} indicates scenarios where abstaining is highly incentivized.}
\label{tab:abstention2}
\begin{tabular}{l*{6}{c}}
& \multicolumn{6}{c}{Abstention Size} \\
\cmidrule(lr){2-7}
& \multicolumn{2}{c}{$k=1$} & \multicolumn{2}{c}{$k=2$} & \multicolumn{2}{c}{$k=3$} \\
\cmidrule(lr){2-3} \cmidrule(lr){4-5} \cmidrule(lr){6-7}
Voter Type & Better & Worse & Better & Worse & Better & Worse \\
\midrule
$a_1 \succ a_2 \succ a_3 \succ a_4$ & 0.023 & 0.021 & 0.051 & 0.022 & 0.063 & 0.045 \\
$a_2 \succ a_1 \succ a_3 \succ a_4$ & 0.000 & 0.010 & 0.000 & 0.010 & 0.004 & 0.017 \\
$a_2 \succ a_3 \succ a_1 \succ a_4$ & 0.006 & 0.006 & 0.004 & 0.016 & 0.008 & 0.039 \\
$a_3 \succ a_2 \succ a_4 \succ a_1$ & 0.002 & 0.012 & 0.002 & 0.031 & 0.002 & 0.055 \\
$a_3 \succ a_4 \succ a_2 \succ a_1$ & 0.000 & 0.002 & 0.000 & 0.010 & 0.004 & 0.012 \\
$a_4 \succ a_3 \succ a_2 \succ a_1$ & \textbf{0.028} & 0.000 & \textbf{0.048} & 0.021 & \textbf{0.075} & 0.023 \\
\end{tabular}
\end{table}

\begin{table}[ht]
\centering
\caption{Fractions of profiles better/worse by Voter Type and Abstention Size for $m=6$. \textbf{Bold} indicates scenarios where abstaining is highly incentivized.}
\label{tab:abstention6}
\begin{tabular}{l*{6}{c}}
& \multicolumn{6}{c}{Abstention Size} \\
\cmidrule(lr){2-7}
& \multicolumn{2}{c}{$k=1$} & \multicolumn{2}{c}{$k=2$} & \multicolumn{2}{c}{$k=3$} \\
\cmidrule(lr){2-3} \cmidrule(lr){4-5} \cmidrule(lr){6-7}
Voter Type & Better & Worse & Better & Worse & Better & Worse \\
\midrule
$a_1 \succ a_2 \succ a_3 \succ a_4 \succ a_5 \succ a_6$ & 0.035 & 0.055 & 0.060 & 0.074 & 0.068 & 0.120 \\
$a_2 \succ a_1 \succ a_3 \succ a_4 \succ a_5 \succ a_6$ & \textbf{0.069} & 0.010 & \textbf{0.109} & 0.011 & \textbf{0.149} & 0.028 \\
$a_2 \succ a_3 \succ a_1 \succ a_4 \succ a_5 \succ a_6$ & 0.042 & 0.029 & 0.056 & 0.048 & 0.090 & 0.100 \\
$a_3 \succ a_2 \succ a_4 \succ a_1 \succ a_5 \succ a_6$ & 0.011 & 0.056 & 0.021 & 0.064 & 0.036 & 0.114 \\
$a_3 \succ a_4 \succ a_2 \succ a_5 \succ a_1 \succ a_6$ & 0.011 & 0.065 & 0.023 & 0.112 & 0.037 & 0.171 \\
$a_4 \succ a_3 \succ a_5 \succ a_2 \succ a_6 \succ a_1$ & 0.014 & 0.031 & 0.027 & 0.093 & 0.043 & 0.142 \\
$a_4 \succ a_5 \succ a_3 \succ a_6 \succ a_2 \succ a_1$ & 0.017 & 0.006 & 0.027 & 0.073 & 0.042 & 0.090 \\
$a_5 \succ a_4 \succ a_6 \succ a_3 \succ a_2 \succ a_1$ & 0.039 & 0.028 & 0.070 & 0.065 & 0.105 & 0.093 \\
$a_5 \succ a_6 \succ a_4 \succ a_3 \succ a_2 \succ a_1$ & \textbf{0.059} & 0.004 & \textbf{0.107} & 0.019 & \textbf{0.153} & 0.019 \\
$a_6 \succ a_5 \succ a_4 \succ a_3 \succ a_2 \succ a_1$ & 0.035 & 0.021 & 0.056 & 0.083 & 0.070 & 0.103 \\
\end{tabular}
\end{table}

\section{Computational Details}

All experiments were run on a CPU without use of GPU. The CPU configuration is given below.
\begin{itemize}
    \item Architecture: x86 64
    \item CPU(s): 16
    \item RAM: 16GB
    \item Thread(s) per core: 2
    \item Core(s) per socket: 8
    \item Socket(s): 1
    \item Processor: AMD Ryzen 7 7700 8-Core Processor (3.80 GHz)
\end{itemize}

All ILPs were solved using the Gurobi solver. We did not run into timeouts for any of the preference profiles. No full solution took more than $90$ seconds to compute, this highest coming for $n=100$ and $m=10$. 

For data generation, we used sets of random seeds to generate the preference profiles for all distributions. 
All code for the ILP and continuum profiles for Section~\ref{sec:theo}'s discussion is at \url{https://github.com/farhadmohsin/gnsp-stv-1d}.
Full set of generated preference profiles, and minimal GNSP cases can be provided upon contacting the corresponding author.

\end{document}